\documentclass[conference]{IEEEtran}

\usepackage{amsmath, amssymb, amsfonts, bm}

\usepackage{graphicx}
\usepackage[caption=false,font=footnotesize]{subfig}
\usepackage{float}      
\usepackage{comment}

\usepackage[vlined, ruled, linesnumbered]{algorithm2e}
\SetKwInput{KwInput}{Input}
\SetKwInput{KwOutput}{Output}
\SetKwRepeat{Do}{do}{until}

\makeatother

\usepackage{xcolor}
\usepackage{listings}
\usepackage{balance}

\lstdefinelanguage{Matlab}{
  morekeywords={break,case,catch,continue,else,elseif,end,for,function,
    global,if,otherwise,persistent,return,switch,try,while},
  sensitive=true,
  morecomment=[l]\%,
  morestring=[m]',
}
\usepackage{amsthm}
\theoremstyle{plain}
\newtheorem{theorem}{Theorem}[section]

\newtheorem{proposition}[theorem]{Proposition}

\theoremstyle{definition}

\makeatletter
\def\th@plain{%
  \thm@notefont{}%
  \itshape
  \normalfont   
}
\makeatother

\usepackage{cite}

\usepackage{tikz}
\usetikzlibrary{positioning,shapes,arrows.meta,calc}
\usetikzlibrary{decorations.pathreplacing}

\tikzset{
  >={Latex[length=3mm]},
  block/.style={draw, thick, rectangle, minimum width=2.8cm, minimum height=1.2cm, align=center},
  filter/.style={draw, thick, trapezium, trapezium left angle=70, trapezium right angle=110,
                 minimum width=3.0cm, minimum height=1.2cm, align=center},
  mixer/.style={circle, draw, thick, minimum size=8mm, inner sep=0pt,
                path picture={
                  \draw[]
                    (-2.4mm,-2.4mm)--(2.4mm,2.4mm)
                    (-2.4mm,2.4mm)--(2.4mm,-2.4mm);
                }},
  conn/.style={thick, -{Latex}}
}

\DeclareFontFamily{U}{mathx}{\hyphenchar\font45}
\DeclareFontShape{U}{mathx}{m}{n}{
      <5> <6> <7> <8> <9> <10>
      <10.95> <12> <14.4> <17.28> <20.74> <24.88>
      mathx10
      }{}
\DeclareSymbolFont{mathx}{U}{mathx}{m}{n}
\DeclareFontSubstitution{U}{mathx}{m}{n}
\DeclareMathAccent{\widecheck}{0}{mathx}{"71}

\DeclareFontFamily{U}{mathx}{\hyphenchar\font45}
\DeclareFontShape{U}{mathx}{m}{n}{<-> mathx10}{}
\DeclareSymbolFont{mathx}{U}{mathx}{m}{n}
\DeclareMathAccent{\widebar}{0}{mathx}{"73}

\newfont{\Bd}{msbm10 at 12 truept}
\newfont{\Sc}{eusm10 at 12 truept}

\newcommand\fH[1]{\sbox0{#1}\dimen0=\ht0 \advance\dimen0 -1ex
  \sbox2{\'{}}\sbox2{\raise\dimen0\box2}%
  {\ooalign{\hidewidth\kern.1em\copy2\kern-.5\wd2\box2\hidewidth\cr\box0\crcr}}}

\def\0{{\bf 0}}
\def\1{{\bf 1}}

\def\c{{\bf c}}

\def\x{{\bf x}}
\def\y{{\bf y}}

\def\B{{\bf B}}

\def\F{{\bf F}}
\def\G{{\bf G}}

\def\0{{\bf 0}}

\def\mbF{{ \mathbb{F} }}

\def\mcA{{ \mathcal{A} }}

\def\mcC{{ \mathcal{C} }}

\def\mcF{{ \mathcal{F} }}

\def\mcH{{ \mathcal{H} }}
\def\mcI{{ \mathcal{I} }}

\def\mcS{{ \mathcal{S} }}
\def\mcT{{ \mathcal{T} }}

\makeatletter
\def\ifundefined{\@ifundefined}
\makeatother

\title{
  Early Selection of Polar Ensemble Decoding with Successive Cancellation Style Component Decoder
  }

\author{Pin-Jing Li and Yu-Chih Huang \\
Institute of Communications Engineering, National Yang Ming Chiao Tung University\\
\{ouo.ee11@nycu.edu.tw, jerryhuang@nycu.edu.tw\}
}

\begin{document}
\maketitle
\begin{abstract}
Parallel ensembles of successive cancellation (SC)-style
decoders improve error-rate performance by decoding multiple
representations of the received vector. Conventional ensembles,
however, keep every member active until completion.
To reduce computational cost while retaining the diversity gain 
of ensemble decoding, we propose early ensemble selection (EES),
which ranks members using intermediate decoding states and 
retains a prescribed numbers of members at each 
predefined selection stage.
Within the EES framework, 
we propose two selection rules, namely 
path-metric-based selection (PM-EES) 
and neural-network-based selection (NN-EES). 
We analytically characterize the computational cost of EES 
through stage-dependent decoder log-likelihood ratio (LLR) 
update counts and provide an information-theoretically motivated 
criterion for choosing selection stages.
Simulations on the \((128,64)\) extended BCH code and the
\(\operatorname{RM}(4,8)\) code show that suitable selection
schedules approach full-ensemble error-rate performance
with fewer decoder LLR updates.
\end{abstract}

\section{Introduction}

Successive cancellation (SC)-style component decoders provide
a common decoding framework across several code families
\cite{geiselhart2021RM_AED,li2026PED,pillet2021UTL,
mazhiming2022BLTA,han2025eBCH_RM}.
Parallel ensembles of these decoders improve error-rate
performance while retaining the latency advantage of
component decoders with small list sizes, at the cost of
running multiple ensemble members.

Conventional parallel ensembles keep every member active
until completion. However, intermediate decoding states may
already provide useful evidence for deciding which members
deserve further computation. This motivates online selection
among partially decoded ensemble members.

Prior work improves ensemble efficiency through offline
construction and online adaptation. Offline methods avoid
redundancy among SC-equivalent members
\cite{pillet2022classification,bioglio2023group}
or select complementary members using simulated decoding
failures~\cite{kestel2023AED_URLLC}, producing a fixed
ensemble independent of the received word.
Reduced-delay automorphism ensemble decoding (AED)
constructs automorphism sets that permit early termination
once partial outputs collectively cover the information
bits~\cite{fominykh2025infoAED_ES}.

Online adaptation has also been studied under different
decoding architectures. Sequential AED launches component
decoders one at a time and uses metrics from completed
decoding attempts to decide whether to activate additional
members~\cite{pillet2026sequentialAED_ES}.
Successive cancellation automorphism list (SCAL) decoding
jointly prunes paths initialized under different
automorphisms within a shared
list~\cite{johannsen2023SCAL}.

For parallel Reed-Muller (RM) code ensembles, quasi-optimal
path convergence (QOPC) uses agreement among current partial
path metrics to trigger member pruning~\cite{tian2025_ES_QOPC}.
Subcode estimate convergence (SEC) uses agreement among
subcode partial path metrics within groups of factor-graph
permutations after half the block has been
decoded~\cite{tian2025_ES_SEC}.
These methods use convergence among members to trigger
selection of a single member for continued decoding,
resulting in computational costs that depend on the
received vector.

A complementary problem is to decide which partially decoded 
members to continue under a prescribed complexity budget.
In this work, we propose early ensemble selection (EES)
for parallel ensembles with SC-style component decoders.
At predefined stages, EES ranks members using intermediate
decoding states and retains a prescribed number of them,
providing explicit control over the decoder log-likelihood ratio (LLR) update count.
The framework preserves independent component decoding
and supports both path-metric-based and learned ranking
policies.
We analyze how the selection stages and retained member
counts determine decoder work, and use the evolution of
decoding uncertainty to motivate stage selection.
Simulations on extended BCH (eBCH) codes using polar ensemble decoding (PED) and RM codes
using AED demonstrate the resulting tradeoff between
error-rate performance and decoder work.

\section{Background}\label{sec:background}
\subsection{Polar Codes and Parallel Ensemble Decoding}

A length-$N=2^m$ polar code or polar subcode encodes
$x=u\G_m$, where
$\G_m=\B_m\F^{\otimes m}$,
$\F=\bigl[\begin{smallmatrix}1&0\\1&1\end{smallmatrix}\bigr]$,
and $\B_m$ is the bit-reversal permutation matrix.
The input vector $u\in\mbF_2^N$ contains $K$ information bits
at positions in the information set $\mcI$.
The remaining positions, indexed by $\mcF=\mcI^c$,
are frozen to zero or constrained to be linear combinations
of preceding information bits.
SC decoding~\cite{arikan2009channel} estimates $u$ sequentially,
whereas successive cancellation list (SCL)
decoding~\cite{tal2015list} retains up to $L$ candidate paths
to mitigate error propagation.

Parallel ensemble decoding, including AED~\cite{geiselhart2021RM_AED,pillet2021UTL}
and PED~\cite{li2026PED},
runs $M$ independent list-$L$ SCL decoders on differently
permuted versions of the received LLR vector.
After all members complete decoding, the most likely candidate
among their outputs is selected.
The ensemble's latency is approximately that of a list-$L$
decoder, while its decoding complexity scales as $MLN\log_2N$.
By trading $M$ against $L$ while keeping $ML$ fixed,
suitable ensembles can match or improve upon the error-rate
performance of a single list-$ML$ SCL
decoder~\cite{geiselhart2021RM_AED,li2026PED}.

At stage $t$, member $i$ extends each surviving path $\ell$
by an admissible bit $b$, with approximate path metric
\begin{equation}
p_{i,t}^{(\ell,b)}
=
p_{i,t-1}^{(\ell)}
+
\mathbf{1}\!\left\{
b\neq\operatorname{HD}\!\left(\lambda_{i,t}^{(\ell)}\right)
\right\}
\left|\lambda_{i,t}^{(\ell)}\right|,
\label{eq:scl_pm_update}
\end{equation}
where $\lambda_{i,t}^{(\ell)}$ is the bit LLR and
$\operatorname{HD}(\lambda)=\mathbf{1}\{\lambda<0\}$.
Both values $b\in\{0,1\}$ are considered at information positions.
At frozen positions, $b$ equals the prescribed value:
zero for static frozen bits or a specified linear combination
of earlier decisions on that path for dynamic frozen bits.
Each member $i$ retains up to $L$ candidates with the smallest
metrics. Denoting the number of retained candidates of member $i$ at stage $t$ by $L_{i,t}\leq L$,
we index their metrics in nondecreasing order as
$p_{i,t,1},p_{i,t,2},\ldots,p_{i,t,L_{i,t}}$.

\subsection{Ensemble Constructions}

Ensemble constructions have been studied for several
code families.
For RM codes, AED uses general affine permutations, allowing all members
to share the same frozen set while decoding differently
permuted versions of the received
vector~\cite{geiselhart2021RM_AED}.
The connection between RM and polar
codes~\cite{arikan2010RM_as_polar} enables the use of
SC-style component decoders for RM-AED.
For polar codes, 
upper triangular linear
and block lower triangular affine structures
can provide SC-compatible automorphism for polar AED 
\cite{pillet2021UTL,mazhiming2022BLTA,pillet2022classification}.
PED uses coordinate permutations
to obtain equivalent polar-subcode representations of a
binary linear block code.
Automorphisms of the original code then generate additional
ensemble members that share the same dynamic frozen
constraints~\cite{li2026PED,trifonov2015subcode,lin2020unidec,lin2025unidec}.
The proposed EES framework applies to these ensembles
by selecting whole members during decoding.


\section{Early Ensemble Selection for Parallel Ensemble Decoding}\label{sec:EES}

The best path metric of each ensemble member can help identify
promising members before decoding completes.
Figure~\ref{fig:pm_justify} shows the probability that retaining
the $M_t$ members with the smallest current best path metrics
preserves at least one member attaining the full ensemble's
minimum final path metric.
This retention probability generally increases as decoding
progresses, indicating that intermediate metrics become more
informative about the final ensemble outcome.
This observation suggests that terminating less promising
members early may reduce decoding complexity with only a
small loss in error-rate performance.

\begin{figure}[!t]
    \centering
    \includegraphics[width=\linewidth]{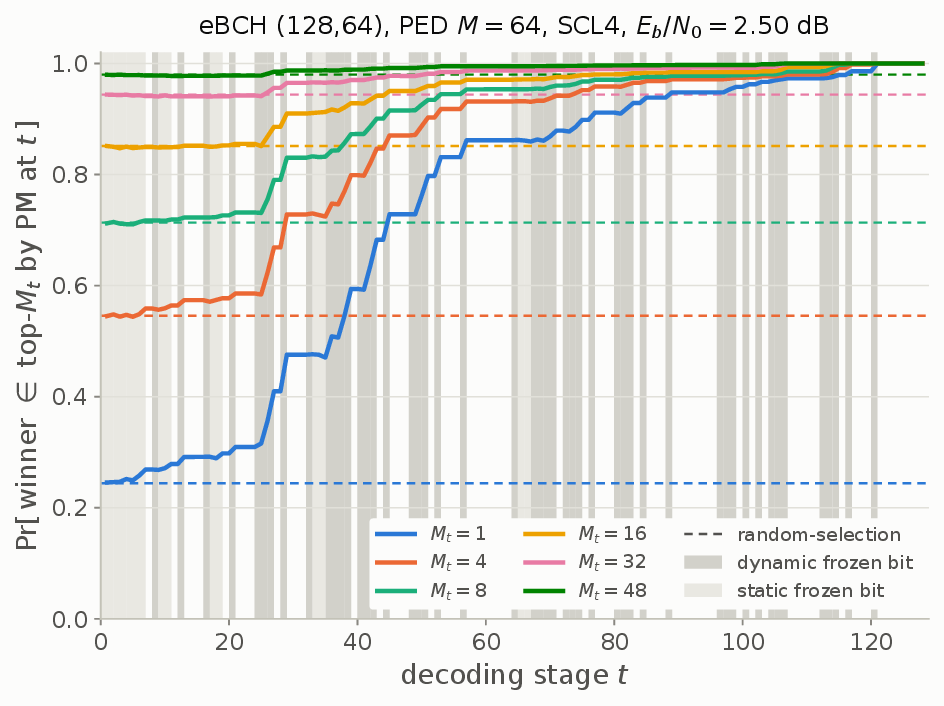}
    \caption{Final-winner retention probability for $(128,64)$ eBCH PED
    with $M=64$ members using SCL ($L=4$) at $E_b/N_0=2.5$ dB.
    Solid curves show the probability that selecting the $M_t$ members
with the smallest current best path metrics after bit $t$ retains
at least one member attaining the full ensemble's minimum final
path metric.
Dashed lines show the corresponding probabilities for uniform
random selection, averaged over frames. These baselines can exceed
$M_t/M$ because multiple members may attain the minimum final
path metric.
    }
    \label{fig:pm_justify}
\end{figure}

For a positive integer $n$, let $[n]:=\{1,\ldots,n\}$.
Consider $M$ ensemble members running independent SCL decoders
with list size $L$ in parallel.
Let $\mcA_t\subseteq[M]$ denote the set of members active
while bit $t$ is decoded, with all members initially active,
i.e., $\mcA_1=[M]$.
Let $\mcT\subseteq[N]$ be a predefined set of selection stages.
At each $t\in\mcT$, after all active members complete bit $t$,
a selector retains promising members based on path-metric
information collected through that bit.

The selector returns a nonempty subset of the active members
\begin{equation}
    \mcS_t=\operatorname{Select}(\mcH_t,\mcA_t),
    \qquad \mcS_t\subseteq\mcA_t,
\end{equation}
where $\mcH_t$ contains the path-metric information available at the
selection stage. Members in $\mcA_t\setminus\mcS_t$ are terminated,
and the active set evolves as
\begin{equation}
    \mcA_{t+1}=
    \begin{cases}
        \mcS_t, & t\in\mcT,\\
        \mcA_t, & t\notin\mcT.
    \end{cases}
    \label{eq:active_ensemble_update}
\end{equation}

We specified the number of retained members at each selection stage
before decoding. Denoting this number by $M_t$, the selector
satisfies $|\mcS_t|=M_t$, with $1\leq M_t\leq|\mcA_t|$.

Algorithm~\ref{alg:EES} summarizes the procedure. 
The following subsections specify the selector and its inputs.

\begin{algorithm}[t]
\caption{Early Ensemble Selection}
\label{alg:EES}
\KwInput{Received vector $\y$, decoder list size $L$, ensemble permutations
$\{\pi_i\}_{i=1}^{M}$, selection stages $\mcT$, retained widths $\{M_t:t \in \mcT\}$,
statistics checkpoints $\mcC=\bigcup_{t\in\mcT}\mcC_t$, selection policy
$\operatorname{Select}(\cdot,\cdot)$}
$\mcA_1\leftarrow[M]$\;
Initialize decoder $i$: $\pi_i(\y),\, p_{i,0,1}=0,\, \forall i\in\mcA_1$\;
\For{$t=1$ \KwTo $N$}{
    Advance all active decoders $i\in\mcA_t$ through bit $t$
    in parallel\;
    \If{$t\in\mcC$}{
        Record $\operatorname{stat}_i(t)$ for each member in $\mcA_t$\;
    }
    \eIf{$t\in\mcT$}{
        Form $\mcH_t$ from the information available through bit $t$\;
        $\mcS_t\leftarrow\operatorname{Select}(\mcH_t,\mcA_t)$\;
        Terminate decoders in $\mcA_t\setminus\mcS_t$\;
        $\mcA_{t+1}\leftarrow\mcS_t$\;
    }{
        $\mcA_{t+1}\leftarrow\mcA_t$\;
    }
}
Undo each member's permutation and select the final candidate from
members in $\mcA_{N+1}$ using the full ensemble's combining rule\;
\KwOutput{Estimated codeword $\hat{\x}$}
\end{algorithm}

\subsection{Selection Based on Path Metrics (PM-EES)}

The selector uses the current best path metric of each active
member,
\(
    \mcH_t=(p_{i,t,1})_{i\in\mcA_t},
\)
without recording metric histories, so $\mcC=\emptyset$.
The resulting selection rule retains the $M_t$ members
with the smallest metrics
\begin{equation}
    \mcS_t=
    \underset{\substack{\mcS\subseteq\mcA_t\\|\mcS|=M_t}}
    {\arg\min}
    \sum_{i\in\mcS}p_{i,t,1}.
    \label{eq:pm_selector}
\end{equation}
Ties are broken in favor of members with smaller indices.

\subsection{Neural-Network-Based Selection (NN-EES)}

PM-EES ranks members using only their current best path
metrics, although each decoder produces a list of path
metrics at every decoding stage.
NN-EES instead assigns each member a learned score based
on its recorded path-metric history and member index.

For each selection stage $t\in\mcT$, let
$\mcC_t\subseteq[t]$ be a predetermined set of statistics
checkpoints that includes the current selection stage,
i.e., $t\in\mcC_t$.
Statistics are recorded at all checkpoints in
$\mcC=\bigcup_{t\in\mcT}\mcC_t$.
After bit $s$, member $i$ holds $L_{i,s}\leq L$ surviving
paths with metrics
$p_{i,s,1}\leq\cdots\leq p_{i,s,L_{i,s}}$.
After completing each checkpoint bit $s\in\mcC$,
each active member $i$ records
\begin{equation}
    \operatorname{stat}_i(s)=\bigl(
        L_{i,s},\;
        p_{i,s,1},\;
        p_{i,s,2}-p_{i,s,1},\;
        \bar p_{i,s},\;
        p_{i,s,L_{i,s}}
    \bigr),
    \label{eq:member_stat}
\end{equation}
where $\bar p_{i,s}$ is the mean of the $L_{i,s}$ path
metrics, and the gap $p_{i,s,2}-p_{i,s,1}$ is set to zero
when $L_{i,s}=1$.

For $i\in\mcA_t$, define the feature vector
\begin{equation}
    \mathbf h_{i,t}=\operatorname{concat}\!\left(
        (\operatorname{stat}_i(s))_{s\in\mcC_t},\,
        \mathbf e_i
    \right),
    \label{eq:model_selector_input}
\end{equation}
where $\mathbf e_i\in\{0,1\}^M$ is the one-hot encoding
of the member index.
The selector input collects the feature vectors of all
active members 
$\mcH_t=(\mathbf h_{i,t})_{i\in\mcA_t}$.

A separate neural network $f_{\theta_t}$ is trained for
each selection stage $t$. 
The network maps each member's feature vector to a
scalar score,
\begin{equation}
    q_{i,t}=f_{\theta_t}(\mathbf h_{i,t}),\qquad i\in\mcA_t,
    \label{eq:model_score}
\end{equation}
and is shared across all active members at that stage.
The selector retains the $M_t$ active members with the
largest scores
\begin{equation}
    \mcS_t=
    \underset{\substack{\mcS\subseteq\mcA_t\\|\mcS|=M_t}}
    {\arg\max}
    \sum_{i\in\mcS}q_{i,t}.
    \label{eq:model_selector}
\end{equation}

Training labels are obtained from the full ensemble.
For each received vector in the training set, all $M$ members
are decoded to completion without selection.
The label of member $i$ is
\begin{equation}
    z_i=
    \begin{cases}
        1, & p_{i,N,1}=\min_{j\in[M]}p_{j,N,1},\\
        0, & \text{otherwise.}
    \end{cases}
    \label{eq:model_label}
\end{equation}
Multiple members may receive positive labels when their
best final path metrics tie for the minimum.

The network is trained by minimizing the binary cross-entropy
between the label $z_i$ and the sigmoid output $\sigma(q_{i,t})$,
summed over all members and all training vectors,
where $\sigma(q)=(1+e^{-q})^{-1}$.
The network architecture and training configuration are given
in Section~\ref{sec:simulation}.


\begin{figure}[!t]
\centering
\resizebox{\columnwidth}{!}{%
\begin{tikzpicture}[
  font=\small,
  eesline/.style={-{Latex[length=1.5mm]}, semithick},
  eesblock/.style={draw, semithick, fill=white, minimum height=6mm,
    inner sep=2pt, align=center},
  eesdec/.style={eesblock, minimum width=11mm},
  eesoff/.style={eesdec, draw=black!30, text=black!45,
    fill=black!3, dashed},
  eesband/.style={draw=black!55, dashed, fill=black!8},
  eesstop/.style={font=\large, inner sep=0pt},
]

\draw[semithick] (1.7,-1.9) -- node[above] {$\y$} (2.5,-1.9);
\draw[semithick] (2.5,0) -- (2.5,-3.8);

\path[eesband] (5.45,0.45) rectangle (5.95,-4.2);
\path[eesband] (8.35,0.45) rectangle (8.85,-4.2);
\node[eesblock, minimum width=15mm, minimum height=4mm] (select1) at (5.7,1.2)
  {Select $t_1$};
\node[eesblock, minimum width=15mm, minimum height=4mm] (select2) at (8.6,1.2)
  {Select $t_2$};
\draw[eesline] ($(select1.west)+(-0.35,0)$) -- (select1.west);
\node[left=0pt] at ($(select1.west)+(-0.35,0)$) {$\mcH_{t_1}$};
\draw[eesline] ($(select2.west)+(-0.35,0)$) -- (select2.west);
\node[left=0pt] at ($(select2.west)+(-0.35,0)$) {$\mcH_{t_2}$};
\draw[eesline] (select1.south) -- (5.7,0.45);
\draw[eesline] (select2.south) -- (8.6,0.45);

\foreach \member/\row in {1/0,2/-1.1,3/-2.2,M/-3.8}{
  \node[eesblock, minimum width=6.5mm] (perm\member) at (3.15,\row)
    {$\pi_{\member}$};
  \node[eesdec] (first\member) at (4.45,\row) {SCL};
  \draw[eesline] (2.5,\row) -- (perm\member);
  \draw[eesline] (perm\member) -- (first\member);
}

\foreach \member/\row in {1/0,3/-2.2,M/-3.8}{
  \node[eesdec] (second\member) at (7.15,\row) {SCL};
  \draw[eesline] (first\member.east) -- (second\member.west);
  \fill (5.7,\row) circle (1.3pt);
}
\draw[semithick] (first2.east) -- (5.6,-1.1);
\node[eesstop] at (5.7,-1.1) {$\times$};
\node[eesoff] at (7.15,-1.1) {off};
\node[eesoff] at (9.85,-1.1) {off};

\foreach \member/\row in {1/0,3/-2.2}{
  \node[eesdec] (last\member) at (9.85,\row) {SCL};
  \node[eesblock, minimum width=7.5mm] (inv\member) at (11.25,\row)
    {$\pi_{\member}^{-1}$};
  \draw[eesline] (second\member.east) -- (last\member.west);
  \fill (8.6,\row) circle (1.3pt);
  \draw[eesline] (last\member) -- (inv\member);
}
\draw[semithick] (secondM.east) -- (8.5,-3.8);
\node[eesstop] at (8.6,-3.8) {$\times$};
\node[eesoff] at (9.85,-3.8) {off};

\foreach \col in {3.15,4.45,7.15,9.85,11.25}{
  \node at (\col,-3.0) {$\vdots$};
}
\node[eesblock, minimum width=12mm, minimum height=46mm]
  (combiner) at (12.75,-1.9) {Final\\combiner};
\draw[eesline] (inv1.east) -- (combiner.west |- inv1.east);
\draw[eesline] (inv3.east) -- (combiner.west |- inv3.east);
\draw[eesline] (combiner.east) -- node[above] {$\hat\c$} (13.9,-1.9);

\node at (4.45,-4.45) {$M$ members};
\node at (7.15,-4.45) {$M_{t_1}$ members};
\node at (9.85,-4.45) {$M_{t_2}$ members};

\end{tikzpicture}%
}
\caption{Ensemble member selection at two stages $t_1<t_2$.
Each row represents one list-$L$ SCL decoder, with successive
blocks showing its decoding intervals.
After all active members complete bit $t_k$, $k\in\{1, 2\}$, 
the selector uses
$\mcH_{t_k}$, including the current best path metrics
$p_{i,t_k,1}$, to retain $M_{t_k}$ members.
Crosses mark example member terminations, and gray blocks indicate
decoding work omitted thereafter.
Surviving members complete decoding, apply their inverse
permutations to the decoded candidates, and pass the candidates
to the final combiner.
The active and selected sets satisfy
$\mcA_{t_1}=[M]\supseteq\mcS_{t_1}=\mcA_{t_2}
\supseteq\mcS_{t_2}$, with $M_{t_k}=|\mcS_{t_k}|$.
The same architecture extends to additional selection stages.}
\label{fig:EES}
\end{figure}
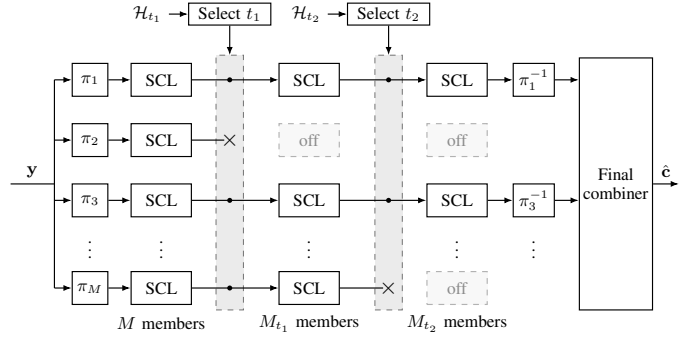

\section{Theoretical Analysis}\label{sec:theorem}
\subsection{Complexity Analysis}
\label{sec:theorem:complexity}

We count scalar LLR updates along the standard SC decoding
traversal~\cite{niu2016SCschedule}.
Let $C_N(b)$ denote the work required
to decode up to bit-$b$ of a length-$N$ code. 

Entering either child of a length-$s$ node requires $s/2$
LLR updates.
Each of the $\log_2N$ tree levels therefore contributes $N$
updates to complete traversal, giving
$C_N(N)=N\log_2N$~\cite{arikan2009channel,tal2015list}.
For $M$ independent SCL decoders, each with list size $L$,
the full-ensemble LLR-update count is $MLN\log_2N$.
For a single selection after bit $t$, the $M-M_t$ terminated
members decode only the first $t$ bits, while the $M_t$
retained members decode all $N$ bits.
The total decoder LLR-update count under this model is
\begin{equation}
  L\bigl[M_t C_N(N)+(M-M_t)C_N(t)\bigr].
  \label{eq:EES_single_stage_cost}
\end{equation}

For selection stages $t_1<\cdots<t_{|\mcT|}$,
let $M_{t_0}=M$ and let $M_{t_j}$ denote the number of
members retained after stage $t_j$.
The $M_{t_{j-1}}-M_{t_j}$ members terminated at that stage
have each completed $t_j$ bits, giving the multistage
LLR-update count
\begin{equation}
  L\left[M_{t_{|\mcT|}}C_N(N)+
  \sum_{j=1}^{|\mcT|}(M_{t_{j-1}}-M_{t_j})C_N(t_j)\right].
  \label{eq:EES_multistage_cost}
\end{equation}
This expression reduces to~\eqref{eq:EES_single_stage_cost}
when $|\mcT|=1$.

\begin{theorem}
For $N\geq2$ a power of two and $1 \leq b \leq N$, 
let $h=N/2$. The cost of decoding the
first $b$ bits satisfies
\begin{equation}
 C_N(b)=
 \begin{cases}
  h+C_h(b), & 0<b\leq h,\\
  N+h\log_2 h+C_h(b-h), & h<b\leq N,
 \end{cases}
 \label{eq:EES_complexity}
\end{equation}
with $C_1(1)=0$.
\end{theorem}
\begin{IEEEproof}
For $0<b\leq h$, the traversal enters the left child at cost $h$ and
decodes $b$ bits within it. For $b>h$, it enters and completes the left
child, then enters the right child and decodes $b-h$ bits there.
Entering both children costs $2h=N$, and completing the left subtree
costs $h\log_2h$, yielding the second branch. At $b=N$, this gives
$N+2h\log_2h=N\log_2N$, consistent with the full traversal.
\end{IEEEproof}

\begin{proposition}
For $b=2^r$, where $r$ is an integer satisfying $0\leq r\leq\log_2N$,
\[
 C_N(b)=N-b+b\log_2b.
\]
\end{proposition}
\begin{IEEEproof}
The descent to the first length-$b$ subtree costs
$N/2+N/4+\cdots+b=N-b$, followed by a complete traversal of that subtree
at cost $b\log_2b$. 
\end{IEEEproof}

\subsection{Choice of Selection Stage $t$}
\label{sec:theorem:t_selection}

Later selection allows more decoding evidence to accumulate,
but the benefit of waiting depends on how informative
the intervening bits are.
To motivate the choice of selection stage, we examine
how decoding uncertainty evolves.

Let $\mcI$ and $\mcF=[N]\setminus\mcI$ denote the information
and frozen sets, and define
$\mcI^{(t)}=\mcI\cap[t]$ and $\mcF^{(t)}=\mcF\cap[t]$.
Let $U_1^N$ denote the polar-input vector and $Y_1^N$ the
corresponding output over 
a binary-input memoryless symmetric (BMS) channel.
Following the auxiliary causal frozen-information model
of~\cite{cocskun2022information}, we quantify the uncertainty
in the information bits processed through stage $t$ by
\begin{equation}
  \bar D_t :=
  H(U_{\mcI^{(t)}}\mid Y_1^N,U_{\mcF^{(t)}}).
\end{equation}
Define the change in uncertainty at stage $t$ as
\(
  \Delta_t := \bar D_t-\bar D_{t-1}.
\)
By the chain rule,
\begin{equation}
  \Delta_t =
  \begin{cases}
    H(U_t\mid Y_1^N,U_1^{t-1}),
      & t\in\mcI,\\
    -I(U_t;U_{\mcI^{(t-1)}}\mid
       Y_1^N,U_{\mcF^{(t-1)}}),
      & t\in\mcF.
  \end{cases}
\end{equation}
Thus, an information bit adds its remaining uncertainty, 
consistent with the branching of candidate decoding paths in SCL.
Revealing a frozen constraint provides a consistency check
on these paths and can reduce uncertainty about preceding
information bits.

These identities motivate placing selection stages immediately
after frozen constraints have been processed, subject to the
decoding complexity budget. 
We adopt this as a heuristic criterion for choosing selection
stages and leave optimal ensemble-selection scheduling
for future work.

For PM-EES, a complementary explanation follows directly from
the approximate path-metric update in~\eqref{eq:scl_pm_update}.
At an information position, every surviving path has a child
with zero metric increment, so each member's best path metric
remains unchanged.
Consequently, positions that are information bits for all
members leave the PM-EES ranking unchanged.
This explains the flat retention probability across shared
information positions in Fig.~\ref{fig:pm_justify} and
motivates selecting after frozen positions, where the
ranking can change.

\section{Simulation Results}\label{sec:simulation}

We evaluate EES over the AWGN channel with BPSK modulation.
We consider the $(128,64)$ eBCH code using
PED-SCL~\cite{li2026PED} with list size $L=4$
in Fig.~\ref{fig:static_EES_tsweep_eBCHm7t10},
and the $\operatorname{RM}(4,8)$ code using
AED-SC~\cite{geiselhart2021RM_AED}
in Fig.~\ref{fig:PM_NN_EES_tsweep_RM_AED}.
Both full ensembles contain $M=64$ members constructed
according to their respective ensemble designs.
Decoder complexity is measured by the scalar LLR-update
count defined in Section~\ref{sec:theorem:complexity}.

For the single-stage NN-EES experiments, a separate scorer
is trained for each code and selection stage.
Each scorer has one hidden layer of eight ReLU units
and a linear output unit.
For $t\geq4$, we use the checkpoint set
\(
  \mcC_t=
  \{4j:1\le j<\lfloor t/4\rfloor\}\cup\{t\}.
\)
Statistics are recorded after each checkpoint bit.
The input dimension is $D_t=5|\mcC_t|+M$.
The scorer has $8D_t+17$ parameters and requires
$40|\mcC_t|+8$ multiplications per member when the
one-hot member index is implemented as a weight lookup.
For each code, we use a training dataset of $96{,}000$
received vectors generated at $E_b/N_0=2.5$ dB,
with member labels defined by~\eqref{eq:model_label}.

The scorer has a small parameter count and requires few
multiplications.
For example, for the $(128,64)$ eBCH code with $M=64$ and $t=16$,
we have $\mcC_{16}=\{4,8,12,16\}$, giving
$D_{16}=5\cdot4+64=84$ and $8\cdot84+17=689$ parameters.
Scoring one member requires $40\cdot4+8=168$ multiplications,
or $64\cdot168=10{,}752$ multiplications for the whole ensemble.
For comparison, the full ensemble with $L=4$ requires
$MLN\log_2N=229{,}376$ LLR updates. The multiplication count
is approximately $4.7\%$ of this update count.
Despite its small size, the scorer improves error-rate
performance over PM-EES at early selection stages,
as shown below.
Further optimization of the network architecture and
evaluation with larger training sets are left for future work.

\begin{figure}[!t]
  \centering
  \includegraphics[width=\columnwidth]
    {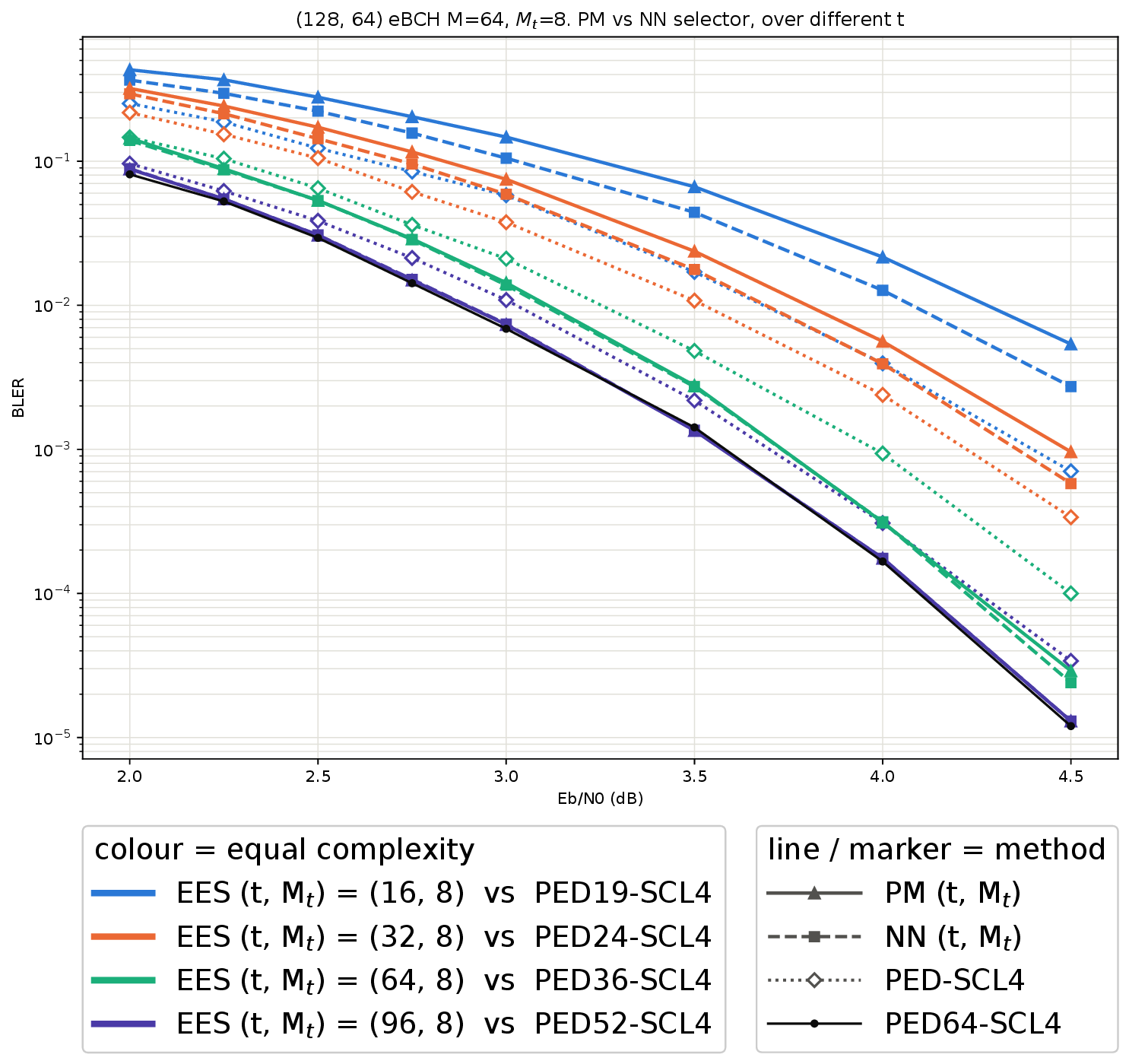}
    \caption{
    Single-stage EES for $(128,64)$ eBCH PED with $M=64$,
    $L=4$, and $M_t=8$, compared with smaller ensembles
    at matched LLR-update counts.
    NN-EES improves error-rate performance over PM-EES
    at $t=16$ and $t=32$.
    At $t=64$, both selectors outperform the matched-cost
    baseline.
    At $t=96$, both approach full-ensemble error-rate
    performance with $186{,}368$ LLR updates,
    approximately $19\%$ fewer than the $229{,}376$ updates
    required by the full ensemble.
}
  \label{fig:static_EES_tsweep_eBCHm7t10}
\end{figure}

\begin{figure}[!t]
  \centering
  \includegraphics[width=\columnwidth]
    {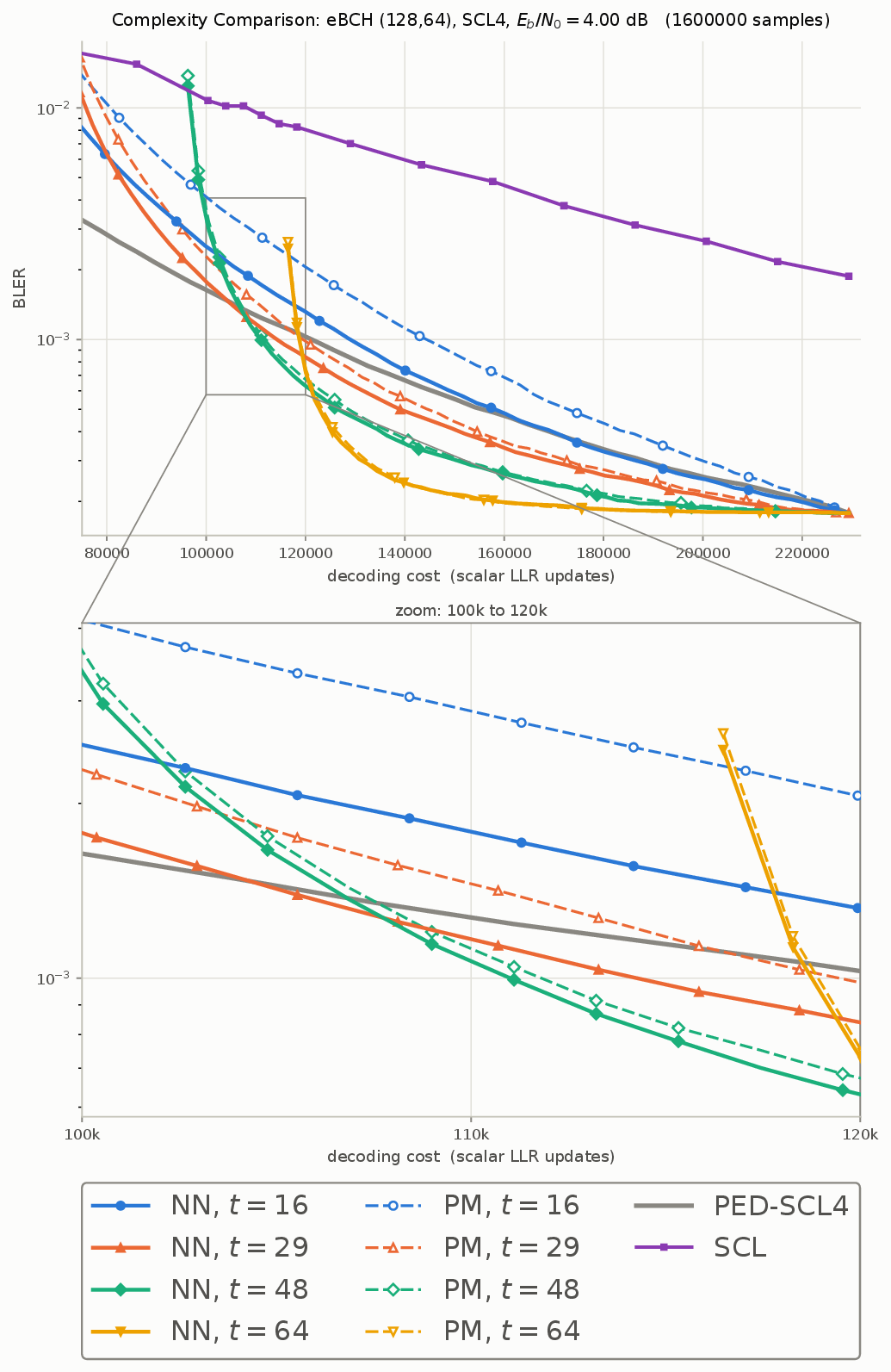}
\caption{
Performance--complexity tradeoff for $(128,64)$ eBCH PED-SCL
at $E_b/N_0=4$ dB, with $M=64$ and $L=4$.
Each EES curve fixes $t$ and varies $M_t$,
with LLR-update counts given by~\eqref{eq:EES_single_stage_cost}.
NN-EES improves the tradeoff at earlier selection stages.
At $t=29$, it outperforms the smaller-ensemble baseline
at a lower LLR-update budget than required by PM-EES.
At $t=48$ and $t=64$, both selectors achieve similar
error-rate performance and outperform the baseline
at higher LLR-update budgets.
}
  \label{fig:complexity_front_2p00}
\end{figure}

Figures~\ref{fig:static_EES_tsweep_eBCHm7t10}
and~\ref{fig:complexity_front_2p00} show the performance of EES
for the $(128,64)$ eBCH code using PED-SCL decoding
at different selection stages $t$.
NN-EES improves error-rate performance over PM-EES
at earlier selection stages.
At later stages, both selectors outperform smaller ensembles
at matched LLR-update counts, demonstrating the benefit of
initially exploring more decoding representations and then
concentrating computation on promising members.

\begin{figure}[!t]
  \centering
  \includegraphics[width=\columnwidth]
    {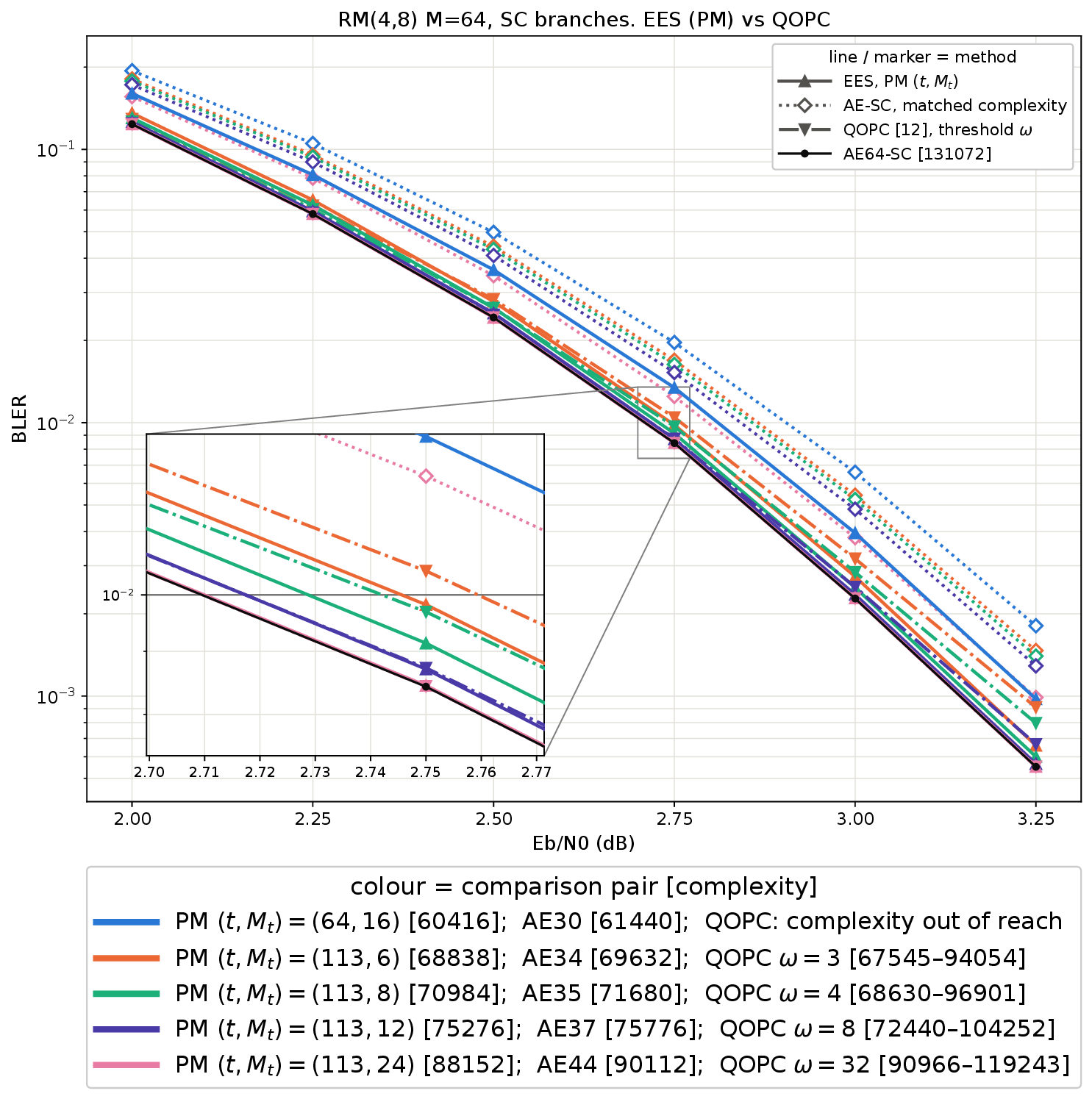}
\caption{
PM-EES and QOPC~\cite{tian2025_ES_QOPC} for
$\operatorname{RM}(4,8)$ AED with $M=64$ SC decoders.
PM-EES settings are denoted by $(t,M_t)$.
For QOPC, selection is triggered when at least $\omega$
members share the unique most frequent partial path metric.
Convergence checks begin at $t=113$ and continue until
selection is triggered or decoding completes.
Colors group settings with similar complexity, including
smaller AED ensembles as complexity baselines.
Bracketed values report decoder LLR-update counts.
For QOPC, these counts depend on the received vector;
the reported ranges give the minimum and maximum mean
counts over the simulated SNR range.
NN-EES is omitted because its performance is nearly
identical to that of PM-EES at these selection stages.
}
  \label{fig:PM_NN_EES_tsweep_RM_AED}
\end{figure}

Figure~\ref{fig:PM_NN_EES_tsweep_RM_AED} compares PM-EES
with QOPC~\cite{tian2025_ES_QOPC} for
$\operatorname{RM}(4,8)$ AED with $M=64$ SC decoders.
For QOPC, the parameter $\omega$ in the legend specifies
the minimum number of members that must share the unique
most frequent partial path metric to trigger selection.
At each verification stage, if this count reaches $\omega$,
QOPC retains the lowest-index member among those sharing
that metric and terminates all other members.

Following~\cite{tian2025_ES_QOPC}, QOPC begins checking
at $t=113$ and continues at subsequent verification stages.
We therefore also evaluate PM-EES at $t=113$ to compare
the two policies at the same initial selection opportunity.
Whereas QOPC adapts its selection stage to the observed
metric agreement and then retains one member, PM-EES
prescribes both the selection stage and the number of
retained members.
Consequently, PM-EES has a fixed decoder LLR-update count,
while QOPC's count depends on the received vector and
its mean varies with SNR.

At $(t,M_t)=(113,24)$, PM-EES achieves essentially
the full-ensemble BLER with $88{,}152$ LLR updates,
approximately $33\%$ fewer than the full ensemble.
QOPC with $\omega=32$ achieves comparable error-rate
performance, but its mean LLR-update count ranges from
$90{,}966$ to $119{,}243$ over the simulated SNR range.
With $M_t=6$ and $M_t=8$, PM-EES provides an approximately
$0.01$ dB gain over the corresponding QOPC settings
while maintaining a fixed LLR-update count.
Earlier PM-EES selection at $(t,M_t)=(64,16)$ further
reduces the count to $60{,}416$, below the minimum
attainable under the configured QOPC verification schedule.
These results demonstrate how prescribing both the
selection stage and the number of retained members
extends the available performance--complexity tradeoff.

\begin{figure}[!t]
  \centering
  \includegraphics[width=\columnwidth]{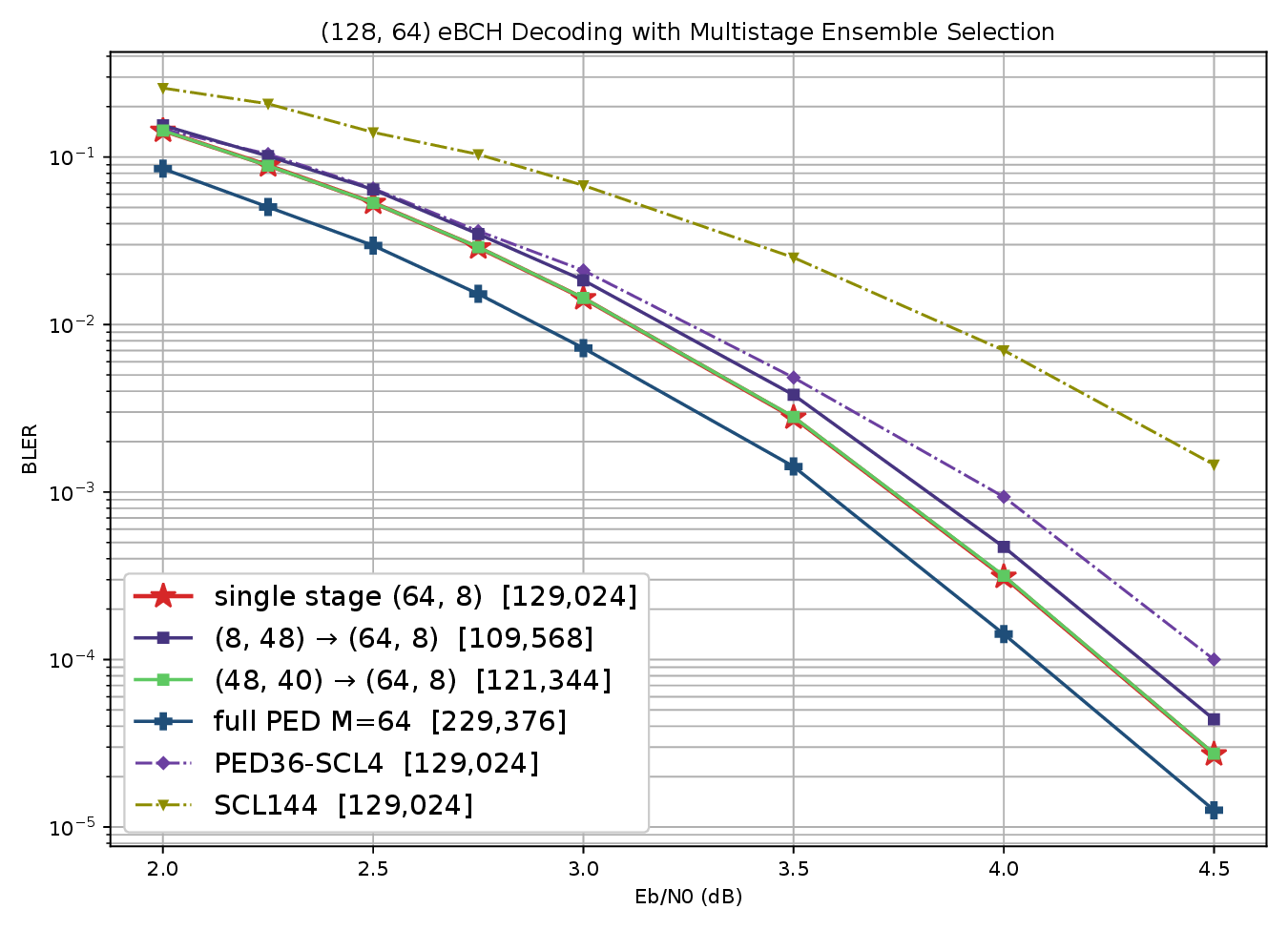}
\caption{
Multistage EES for $(128,64)$ eBCH PED with $M=64$
and $L=4$.
A schedule is denoted by
$(t_1,M_{t_1})\rightarrow(t_2,M_{t_2})$.
The $(48,40)\rightarrow(64,8)$ schedule uses PM-EES
at both stages, while $(8,48)\rightarrow(64,8)$ uses
NN-EES followed by PM-EES.
The NN-EES scorer uses the architecture and training
setup described in this section.
Both schedules closely track the error-rate performance
of single-stage EES at $(t,M_t)=(64,8)$.
Their LLR-update counts are $121{,}344$ and $109{,}568$,
respectively, approximately $6\%$ and $15\%$ fewer than
the $129{,}024$ updates required by the single-stage baseline.
Square brackets report LLR-update counts calculated
from~\eqref{eq:EES_multistage_cost}.
}
  \label{fig:multistage}
\end{figure}

Figure~\ref{fig:multistage} illustrates multistage EES,
which provides additional flexibility in the
performance--complexity tradeoff.
PM-EES or NN-EES can be used at each selection stage,
depending on the available offline training resources.

\section{Conclusions}\label{sec:conclude}

We proposed early ensemble selection (EES) for parallel
ensembles with SC-style component decoders, using partial
path metrics and their statistics to retain a prescribed
number of members at predefined stages.
Simulations on eBCH PED and RM AED show that suitable
selection schedules approach full-ensemble error-rate
performance with fewer decoder LLR updates.
NN-EES improves error-rate performance over PM-EES at
earlier selection stages, while the two selectors achieve
similar performance at later stages.
Future work includes joint optimization of selection stages
and retained member counts, adaptive ensemble selection,
joint path pruning across ensemble members, and further
information-theoretic analysis of EES.

\bibliographystyle{IEEEtran}
\bibliography{ref}  

\end{document}